\documentclass{amsart}
\usepackage{latexsym}
\usepackage{amssymb}
\usepackage{amsmath}
\usepackage{amsfonts}
\usepackage{enumerate}
\usepackage{graphicx}

\newcommand{\bmat}{\left[\begin{matrix}}
\newcommand{\emat}{\end{matrix}\right]}
\usepackage[latin1]{inputenc} 
\usepackage{amsthm}
\usepackage{amsxtra}
\usepackage{amscd}
\usepackage{setspace}   
\usepackage{mathrsfs}   

\newtheorem{theorem}{Theorem}
\newtheorem{proposition}[theorem]{Proposition}
\newtheorem{lemma}[theorem]{Lemma}

\theoremstyle{remark}
\newtheorem*{remark}{Remark}

\theoremstyle{definition}

\newcommand{\Z}{\mathbb{Z}}

    \title{A stochastic variant of the abelian sandpile model}
    \author{Seungki Kim and Yuntao Wang}

\begin{document}
\maketitle

\begin{abstract}
We introduce a natural stochastic extension, called \emph{SSP}, of the abelian sandpile model(ASM), which shares many mathematical properties with ASM, yet radically differs in its physical behavior, for example in terms of the shape of the steady state and of the avalanche size distribution. We establish a basic theory of SSP analogous to that of ASM, and present a brief numerical study of its behavior.

Our original motivation for studying SSP stems from its connection to the LLL algorithm established in another work by the authors \cite{DKTW}. The importance of understanding how LLL works cannot be stressed more, especially from the point of view of lattice-based cryptography. We believe SSP serves as a tractable toy model of LLL that would help further our understanding of it.
\end{abstract}

\section{Introduction}

\subsection{Overview}


Let us start by recalling ASM (\cite{BTW87}, \cite{Dhar90}, also see \cite{Dhar06}) on a one-dimensional lattice. Let $A_L$ be the cycle graph whose vertices $V(A_L)$ consist of $L+1$ elements, say $V(A_L) = \{v_1, \ldots, v_{L+1}\}$, and the (undirected) edges are defined by $E(A_L) = \{(v_i, v_{i+1}): i = 1, \ldots, L\} \cup \{(v_{L+1}, v_1)\}$. We designate $v_{L+1}$ as the \emph{sink}. Fix two positive integers $T$ and $I$, preferably $T \geq 2I$.

We temporarily refer to $I$ as the ``toppling strength'' and $T$ as the \emph{threshold}. That is, toppling at site $v_i$ is to subtract $2I$ grains from the pile on $v_i$ and add $I$ grains to each of its neighbors. A configuration on $A_{L}$ is called \emph{stable} if and only if all non-sink vertices have $< T$ grains. Fix another positive integer $J$ coprime to $I$, preferably $J \approx I$. One then considers the following Markov chain on the space of the stable configurations of ASM: given a stable configuration on $A_L$, obtain the next stable configuration by first adding $J$ grains to any randomly chosen non-sink site, and then stabilizing the resulting configuration.\footnote{The only reason to impose the condition $(J, I) = 1$ is to prevent the pile heights at each site from being concentrated on a select few congruence classes modulo $I$; it is not so much an essential condition as a cosmetic one.}

On the other hand, consider the following simple stochastic variant of ASM, which we named \emph{SSP}. SSP is exactly the same as ASM except for the toppling procedure: first one samples $\gamma$ uniformly from $\{1, 2, \ldots, 2I\}$ (so that $\gamma$ has mean $\approx I$) and then subtract $2\gamma$ grains from site $v_i$ and add $\gamma$ to each of the two neighboring sites. One can of course associate a Markov chain to SSP analogously to the above discussion: add $I$ grains to a random site and then stabilize.

It is rather surprising that this simple and natural extension of ASM has never been considered in the literature so far. In fact, our motivation for studying SSP does not come from physics, but from a study of the practical behavior of the LLL algorithm \cite{LLL82}. The LLL algorithm is one of the fundamental tools in computational mathematics used for lattice basis reduction, with numerous applications to number theory and cryptography --- see the book \cite{LLLbook} for an extensive treatment on the subject. Despite its celebrated status, much of its behavior in practice has been shrouded in complete mystery. For example, the most well-known problem concerns the quantity called the \emph{root Hermite factor}(RHF), defined as

\begin{equation} \label{eq:rhf}
\log\mbox{(RHF)} = \frac{1}{(L+1)^2}\sum_{i=1}^{L}(L+1-i)r_i,
\end{equation}
where $r_i$ is the pile height at $v_i$, in sandpile language. It is a theorem from \cite{LLL82} that the RHF of LLL has a sharp upper bound by $(4/3)^{1/4} \approx 1.075$, but empirically one observes RHF $\approx 1.02$ most of the time --- see \cite{NS06} for details. This phenomenon has been well-known since the birth of LLL in 1982, yet there has been not even a vague heuristic as to why this must happen.

Recently, we argued in \cite{DKTW} that LLL behaves like a sandpile model, and that this idea explains much of its practical behavior all at once, in particular this RHF problem. The remaining difficulty is that LLL as a sandpile model is nonabelian, which is difficult to study analytically. Hence we invented SSP, the closest abelian version of LLL, to be used as a toy model.

Figure \ref{fig:lllssp} compares the empirical average ``steady state density'' of LLL and SSP. The resemblance therein should come as surprising, since the two algorithms originate from very different fields of study, lattice reduction and statistical physics, respectively. We hope that the present paper helps build a bridge between them, to the benefit of both areas.

\begin{figure}
\includegraphics[scale=0.45]{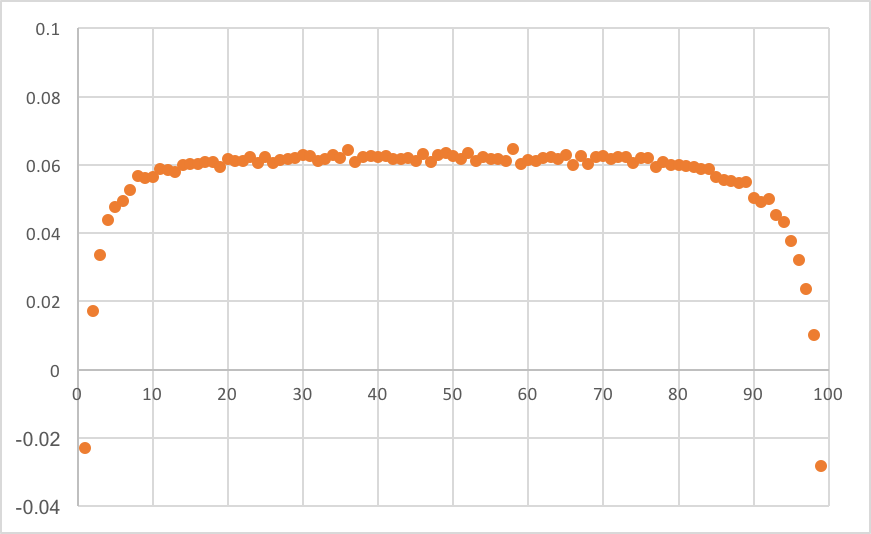}
\includegraphics[scale=0.45]{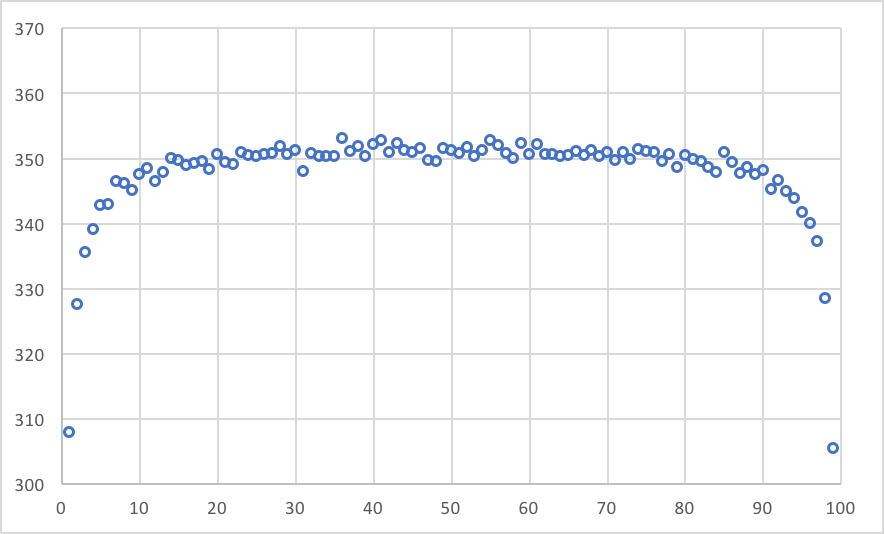}
\caption{The empirical average ``steady state density'' of LLL (left) and SSP (right), with $L = 100$. $x-$ and $y-$ axes represent the site index and the height, respectively. It would be ideal to take a higher $L$, but then running LLL would become cumbersome, which has complexity $O(L^5)$.}
\label{fig:lllssp}
\end{figure}



\subsection{In this paper}

In the next section, we develop a basic mathematical theory of SSP. As in the case of ASM, we impose a monoid structure on the set of stable configurations, which we show is abelian (Theorem \ref{thm:monoid}), and we prove the existence of the unique steady state (Theorem \ref{thm:steady}). Many of our results can be proved using the operator algebra method e.g. as in \cite{SD09}, but here we introduce a different idea, which is useful for picturing the ``shape'' of the steady state. With its help we are able to prove that the average $\log(\mathrm{RHF})$ of SSP over the steady state is strictly bounded away from the worst-case by a constant independent of system size $L$ (Proposition \ref{prop:rhf}). This is precisely what one wants to show with LLL, and hopefully our idea generalizes to this case, the only obstacle being its non-abelian toppling rule.

In Section 3, we provide some preliminary numerical studies on the statistical behavior of SSP, concerning its avalanche size distribution and its steady state. Though we still need a more extensive investigation, it seems that SSP has much in common with other one-dimensional stochastic models, such as the abelian Manna model and the Oslo rice-pile model. From the LLL perspective, it is especially interesting that both the Oslo and the Manna models exhibit the same boundary phenomenon as can be seen in Figure \ref{fig:lllssp} (see \cite{GDM16} and \cite{SD09}).



\subsection{Acknowledgment}

We thank Deepak Dhar, Phong Nguyen, and Su-Chan Park for helpful comments and suggestions.


\section{Mathematical properties of SSP}

\subsection{Setup}

In the literature on sandpiles, there seem to exist two sets of notations, one that is physics-oriented and one that is mathematics-oriented --- heights versus configurations, sites versus vertices, relaxations versus stabilizations, and so on. This paper mostly employs the latter, and there are several words that we made up ourselves. In any case, we provide the definition of every term that we use below.

Like ASM, SSP is played on a finite undirected graph $G = (V,E)$, equipped with a designated vertex called \emph{sink}, which we denote here by $s$. Throughout this paper, we assume that the graph obtained from $G$ by removing $s$ is connected. In addition, let $\wp$ be a probability mass function supported on a finite subset of $\mathbb{Z}_{>0}$, and let $T$ be an integer satisfying $T \geq \max\{i: i \in \mathrm{supp}\, \wp\}$, which we sometimes refer to as the \emph{threshold}.

A \emph{pure configuration} is a function $c: V \backslash \{s\} \rightarrow \mathbb{Z}$. A \emph{(mixed) configuration} is a formal finite sum of pure configurations of form $\sum p_i [c_i]$, where $p_i$ are positive real numbers such that $\sum p_i = 1$, and $c_i$ are pure configurations. The role the bracket notation $[.]$ here is to clarify that this is a \emph{formal} sum. We define configurations this way, since toppling in SSP leads to multiple probabilistic results, as will be explained below.

A pure configuration $c$ is called \emph{stable} if $c(v) < T$ for all $v \in V \backslash \{s\}$, and called \emph{nonnegative} if $c(v) \geq 0$ for all $v \in V \backslash \{s\}$. A configuration is called \emph{stable} or \emph{nonnegative} if it is a formal sum of stable or nonnegative pure configurations, respectively. 

Toppling a pure configuration $c$ at a non-sink vertex $v$ goes as follows: i) sample $\gamma$ from $I$ according to $\wp$, ii) subtract $\gamma \cdot \mathrm{deg}(v)$ from $c(v)$ iii) for each $w \sim v$ add $\gamma \cdot \mathrm{deg}(v,w)$ to $c(w)$, where $\mathrm{deg}(v,w)$ is the number of edges connecting $v$ and $w$. Denote the outcome by $T_v(c)$. If $c(v) \geq T$, we call this toppling \emph{legal}. It is clear that, unless $\wp$ is supported on a singleton set, toppling a pure configuration results in a mixed configuration.

A \emph{choice algorithm} is a map $\mathcal{A}: \{\mbox{unstable pure configurations}\} \rightarrow V \backslash \{s\}$, such that $c(\mathcal{A}(c)) \geq T$. Fixing a choice algorithm $\mathcal{A}$, we can define what it means to topple a mixed configuration $C = \sum p_i[c_i]$. For a pure configuration $c$, define

\begin{equation*}
T(c) := T_{\mathcal{A}(c)}.
\end{equation*}

Then we define $T(C) := \sum p_i T(c_i)$, the outcome of toppling $C$ once. The \emph{stabilization} $C^\circ$ of $C$ is the outcome of toppling $C$ repeatedly until it becomes stable. Later we will show that $C^\circ$ is independent of $\mathcal{A}$.

\begin{remark}
It is possible to extend the definition of a choice algorithm to a map into the set of probability mass functions on $V\backslash\{s\}$, so that $\mathcal{A}(c)(v) > 0$ only if $c(v) \geq T$, and
\begin{equation*}
T(c) := \sum_{v \in V\backslash\{s\}} \mathcal{A}(c)(v)T_v(c).
\end{equation*}

All the results in this paper carry over to this generalization.
\end{remark}




As in the case of ASM, we can define an operation $\oplus$ on the set of nonnegative stable configurations. If $c$ and $d$ are pure stable configurations, we define $c \oplus d = (c+d)^\circ$, where $c+d$ is defined so that $(c+d)(v) = c(v) + d(v)$. In general, if $C = \sum p_ic_i$ and $D = q_jd_j$, then we define $C + D = \sum_{i,j} p_iq_j[c_i + d_j]$ and accordingly $C \oplus D = (C + D)^\circ = \sum_{i,j} p_iq_j(c_i \oplus d_j)$.

Because we assumed that $T \geq \max\{i: i \in I\}$, for any nonnegative configurations $c$ and $d$, $c \oplus d = (c+d)^\circ$ is nonnegative as well. Later, we will show that $\oplus$ is abelian and associative, that the set of nonnegative stable configurations form a semigroup under this operation, whose minimal ideal consists of a single element which is the steady state.

\subsection{Stabilization is independent of choice algorithms}

We start by defining what we temporarily call an \emph{increment stack}. An increment stack consists of $|V \backslash \{s\}|$ sequences, each of which is indexed by an element of $V \backslash \{s\}$, whose entries are chosen from the elements of $\mathrm{supp}\,\wp$; thus it has the form $S = (\gamma_{v,i})_{v \in V \backslash \{s\} \atop i \geq 1}$, with each $\gamma_{v,i} \in \mathrm{supp}\, \wp$. The set of all increment stacks is given a topology and a measure as follows: the cylinder sets, i.e. sets of form $Y(m, (v_k, i_k, p_k)_{k=1}^{m}) = \{(\gamma_{v,i})_{v \in V \backslash \{s\} \atop i \geq 1} : \gamma_{v_k,i_k} = p_k, k = 1, \ldots, m\}$, form the base for the closed sets, and the measure $\mu$ is defined by $\mu(Y(m, (v_k, i_k, p_k)_{k=1}^{m})) = \prod_{k=1}^m \wp(p_k)$. This notion simulates the randomness in the amount being toppled during a stabilization process: SSP is equivalent to first sampling an increment stack $S$ according to $\mu$, and if one must topple at $v$ for the $i$-th time, remove $\gamma_{v,i} \cdot \mathrm{deg}(v)$ unit of sand from $v$ and distribute it equally among the neighbors of $v$ (respecting the multiplicities).

For a function $k: V\backslash\{s\} \rightarrow \Z_{\geq 0}$, a \emph{substack} $S(k(v))$ of an increment stack $S = (\gamma_{v,i})_{v \in V \backslash \{s\} \atop i \geq 1}$ is a finite-length sequence of form $(\gamma_{v,i(v)})_{v \in V \backslash \{s\} \atop 1 \leq i(v) \leq k(v)}$. We define the following partial order on the set of substacks: $S(k(v)) \leq S(l(v))$ if $k(v) \leq l(v)$ for $v \in V\backslash\{s\}$. Given a pure configuration $c$, if for every $v \in V\backslash\{s\}$
\begin{equation*}
c(v) - \sum_{i=1}^{k(v)} \gamma_{v,i} \cdot \mathrm{deg}(v) + \sum_{w \sim v} \sum_{i=1}^{k(w)} \gamma_{w,i}\mathrm{deg}(w,v) < T,
\end{equation*}
or equivalently, if toppling, legally or illegally, according to $S(k(v))$ stabilizes $c$ --- then $S(k(v))$ is called a \emph{stabilizing substack} of $c$.

\begin{lemma} \label{lemma:mss}
For any given increment stack $S$ and pure configuration $c$, there exists a unique minimal stabilizing substack of $c$.
\end{lemma}
\begin{proof}
Suppose $P := S(k(v))$ and $Q := S(l(v))$ are two distinct minimal stabilizing substacks. Define $r(v) = \min(k(v),l(v))$, and let $R := S(r(v))$ be another substack; by assumption $R < P$ and $R < Q$. Then, for every non-sink $v$,
\begin{equation*}
c(v) - \sum_{i=1}^{r(v)} \gamma_{v,i} \cdot \mathrm{deg}(v) + \sum_{w \sim v} \sum_{i=1}^{r(w)} \gamma_{w,i}\mathrm{deg}(v,w) < T.
\end{equation*}

Therefore $R$ is a stabilizing substack, contradicting the minimality of $P$ and $Q$.
\end{proof}

Given an increment stack $S$ and a pure configuration $c$, a choice algorithm $\mathcal{A}$ induces a stabilizing substack $S(k_\mathcal{A}(v))$, where $k_\mathcal{A}(v)$ is the number of times $\mathcal{A}$ would topple on $v$ until it stabilizes $c$.

\begin{proposition} \label{prop:unique_stab}
Any choice algorithm induces the minimal stabilizing substack.
\end{proposition}
\begin{proof}
Write $M := S(m(v))$ for the minimal stabilizing substack, and $N := S(k_\mathcal{A}(v))$ for the substack induced by a choice algorithm $\mathcal{A}$. Imagine $\mathcal{A}$ toppling an input configuration $c$, one vertex at a time, and consider a situation in which $\mathcal{A}$ has just toppled on $v_1 \in V$ for $m(v_1)$ times, $v_1$ being the first vertex at which this occurs. Observe that, at this point, the configuration is stabilized at $v_1$. Hence $\mathcal{A}$ cannot topple more on $v_1$, at least not until it topples on some neighborhood $w$ of $v_1$ more than $m(w)$ times. 

Suppose $v_2$ is the second vertex at which $\mathcal{A}$ topples on $v_2$ for $m(v_2)$ times. By the same argument as above, $\mathcal{A}$ cannot topple on $v_2$ more than $m(v_2)$ times until it topples on some $w \sim v_2$ more than $m(w)$ times. Repeating, we see that on no vertices $v$ can $\mathcal{A}$ topple more than $m(v)$ times. This proves that $N \leq M$, but since $M$ is minimal we must have $M = N$.
\end{proof}

An immediate consequence of Proposition \ref{prop:unique_stab} is that, for any configuration $C$, pure or non-pure, its stabilization $C^\circ$ is independent of the choice algorithm $\mathcal{A}$. Indeed, $C^\circ$ is determined solely by the measure $\mu$ on the set of all increment stacks, and the minimal stabilizing substack of each increment stack. Therefore, the notion of a choice algorithm is somewhat superfluous from a theoretical viewpoint.

\subsection{$\oplus$ is abelian and associative}

It is clear by definition that $\oplus$ is abelian. We will now show that $\oplus$ is associative as well. Since $(F \oplus G) \oplus H = ((F + G)^\circ + H)^\circ$ and $F \oplus (G \oplus H) = (F + (G + H)^\circ)^\circ$, it suffices to show that

\begin{proposition} \label{prop:assoc}
For nonnegative configurations $C$ and $D$,
\begin{equation} \label{eq:assoc}
(C^\circ + D)^\circ = (C + D)^\circ,
\end{equation}
or equivalently,
\begin{equation*}
C^\circ \oplus D = C \oplus D.
\end{equation*}
\end{proposition}
\begin{proof}
Assume first that $C$ and $D$ are pure configurations. Fix an increment stack $S = (\gamma_{v,i})_{v \in V \backslash \{s\} \atop i \geq 1}$. Let $S(k(v))$ and $S(l(v))$ be the minimal stabilizing substacks for $C$ and $C+D$, respectively. Note $S(k(v)) \leq S(l(v))$ --- this is where we use the nonnegativity of $C$ and $D$.

Consider the left-hand side of \eqref{eq:assoc}, namely $(C^\circ + D)^\circ$. \emph{A priori}, we need two increment stacks for each of the two stabilizations: $T_1$ for the inner term $C^\circ$ and $T_2$ for the outer one. Write $S' = (\gamma_{v,i + k(v)})_{v \in V \backslash \{s\} \atop i \geq 1}$. It suffices to prove that
\begin{align*}
&\mu(\{T_1: T_1(k(v)) = S(k(v))\}) \cdot \mu(\{T_2: T_2(l(v)-k(v)) = S'(l(v)-k(v))\}) \\
&= \mu(\{R: R(l(v)) = S(l(v))\}),
\end{align*}
or, in other words, that carrying out both stabilizations on the left-hand side of \eqref{eq:assoc} using only one increment stack $S$ induces no distortion on the measure $\mu$. But this is clear from the definition of $\mu$.

It remains to consider the case in which $C$ and $D$ are not necessarily pure configurations. Write $C = \sum p_i[c_i]$ and $D = \sum q_j[d_j]$. Then
\begin{equation*}
C^\circ \oplus D = \sum p_iq_j(c_i^\circ \oplus d_j) = \sum p_iq_j(c_i \oplus d_j) = C \oplus D,
\end{equation*}
as desired.
\end{proof}

\begin{remark}
If the nonnegativity assumption is not present, Proposition \ref{prop:assoc} fails. In fact, it fails for ASM as well. 
\end{remark}

Thanks to our results so far, we have the following
\begin{theorem} \label{thm:monoid}
The set of all nonnegative stable configurations of a SSP forms a commutative monoid under $\oplus$.
\end{theorem}

\subsection{The existence of unique steady state}

Denote the monoid in Theorem \ref{thm:monoid} by $M$. A \emph{steady state} $S \in M$ is an element such that
\begin{equation*}
S \oplus C = S \mbox{ for all $C \in M$.}
\end{equation*}

It is clear from this definition that a steady state, if exists, is unique.

\begin{theorem} \label{thm:steady}
Assume the greatest common divisor of the elements of $\mathrm{supp}\, \wp$ equals $1$. Then $M$ has a steady state.
\end{theorem}

There is no loss of generality incurred by the assumption in the theorem, since if we divide all elements of $\mathrm{supp}\, \wp$ by the g.c.d. then we obtain essentially the same model.

Our proof of Theorem \ref{thm:steady} relies almost entirely on the following lemma.

\begin{lemma} \label{lemma:steady}
Let $p: \{1, 2, \ldots, n\} \rightarrow \mathbb{R}_{\geq 0}$ such that $\sum p(i) = 1$, and such that the g.c.d. of the elements of $\mathrm{supp}\, p$ equals $1$. Consider the $n \times n$ matrix
\begin{equation*}
A =
\begin{bmatrix}
p(1) & 1 & 0 & 0 & \dots  & 0 \\
p(2) & 0 & 1 & 0 & \dots  & 0 \\
p(3) & 0 & 0 & 1 & \dots  & 0 \\
\vdots & \vdots & \vdots & \vdots & \ddots & \vdots \\
p(n) & 0 & 0 & \dots & \dots  & 0
\end{bmatrix}
\end{equation*}
with the first column filled with $p(i)$, the superdiagonal with $1$, and $0$ elsewhere. Then its (complex right) eigenvalue with the largest modulus is 1, which has multiplicity 1, with the eigenvector
\begin{equation} \label{eq:ev}
\begin{bmatrix}
r(1) \\
r(2) \\
\vdots \\
r(n)
\end{bmatrix}
\end{equation}
where $r(k) = \sum_{i \geq k} p(i)$.
\end{lemma}
\begin{proof}
The characteristic equation of $A$ equals
\begin{equation*}
\chi_A(\lambda) = \lambda^n - p(1)\lambda^{n-1} - \ldots - p(n-1)\lambda - p(n).
\end{equation*}

We first show that $1$ is an eigenvalue of multiplicity one. Certainly $\chi_A(1) = 0$. Also since
\begin{equation*}
\chi'_A(\lambda) = n\lambda^{n-1} - (n-1)p(1)\lambda^{n-2} - \ldots - p(n-1)
\end{equation*}
and $n > (n-1)p(1) + \ldots p(n-1)$, $\chi'_A(1) \neq 0$. It is straightforward to check that \eqref{eq:ev} is a solution to $Ax = x$. This proves the latter two claims of the lemma.

If $|\lambda| > 1$, then $|\lambda|^n > \sum p(i)|\lambda|^{n-i}$, so it cannot be a root of $\chi_A$. Therefore it suffices to show that $1$ is the only complex eigenvalue with modulus $1$. Suppose $|\lambda| = 1$ and $\chi_A(\lambda)$ = 0. Then for all $i \in \mathrm{supp}\, p$, $\lambda^{n-i} = \lambda^n$ must hold. In other words, $\lambda$ is an $i$-th root of unity for all $i \in \mathrm{supp}\, p$. But since the g.c.d. of the elements of $\mathrm{supp}\, p$ equals 1 by assumption, $\lambda = 1$ is forced.
\end{proof}

An immediate consequence of Lemma \ref{lemma:steady} is that, for any $x \in \mathbb{R}^n$, $A^kx$ approaches a constant multiple of \eqref{eq:ev} as $k \rightarrow \infty$. This is how we use Lemma \ref{lemma:steady} in the argument below.

\begin{proof}[Proof of Theorem \ref{thm:steady}]
Write $V \backslash \{s\} = \{v_1, v_2, \ldots, v_L\}$, and let $n$ be the smallest positive integer such that $\mathrm{supp}\, \wp \subseteq \{1, \ldots, n\}$. Define $\tau_{i}$ to be the toppling operator at $v_i$ with increment $1$ (hence $\tau_i$ coincides with the toppling operator of the standard ASM on $G$), so that $T_{v_i} = \sum_j \wp(j)\tau_i^j,$ where $T_{v_i}$ is the toppling operator of SSP.

Define $r: \{0, \ldots, n-1\} \rightarrow \mathbb{R}$ by $r(k) = C\sum_{i > k}^{n} \wp(i)$, where the constant factor $C$ is set so that $\sum r(k) = 1$, that is,
\begin{equation*}
C = \left(\sum_{i=1}^n \sum_{j=i}^n \wp(j)\right)^{-1}.
\end{equation*}

For any pure configuration $c$, define
\begin{equation*}
S(c) := \sum_{k_1 = 0}^{n-1} \cdots \sum_{k_L = 0}^{n-1} \left(\prod_{j=1}^L r(k_j)\tau_j^{k_j}\right)[c].
\end{equation*}

One can think of $S(c)$ as a distribution supported on $\prod_{i=1}^L k_i$ points on the space of (pure) configurations, which turn out to form the shape of a parallelepiped; see Figure \ref{fig:pic} for an illustration in case $L = 2, \mathrm{supp}\, \wp = \{1, 2, 3, 4, 5\}$. The point on the upper-right corner corresponds to $[c]$, and other points are obtained by toppling $[c]$ in various directions; see again Figure \ref{fig:pic} for an illustration.

Write $I(\cdot)$ for the indicator function of the statement inside the parenthesis, i.e. if it is satisfied, $I(\cdot) = 1$, otherwise zero. From Lemma \ref{lemma:steady}, it follows that
\begin{align}
&\sum_{k_1 = 0}^{n-1} \cdots \sum_{k_L = 0}^{n-1} \left(\prod_{j=1}^L r(k_j)\tau_j^{k_j}\right)T_{v_i}^{I(k_i = 0)}[c] \notag \\
&= \sum_{k_1 = 0}^{n-1} \cdots \sum_{k_L = 0}^{n-1} \left(\prod_{j=1}^L r(k_j)\tau_j^{k_j}\right)\tau_i[c] \notag \\
&= S(\tau_ic). \label{eq:commute}
\end{align}

This means that, if we ``push'' the parallelepiped-shaped distribution in the direction of $\tau_{i}$ by applying $T_{v_i}$ to the face of the parallelepiped corresponding to that direction, we obtain the same distribution on the parallepiped, except that the upper-right corner of the parallelepiped moves from $[c]$ to $[\tau_ic]$.

\begin{figure}
\includegraphics[scale=0.5]{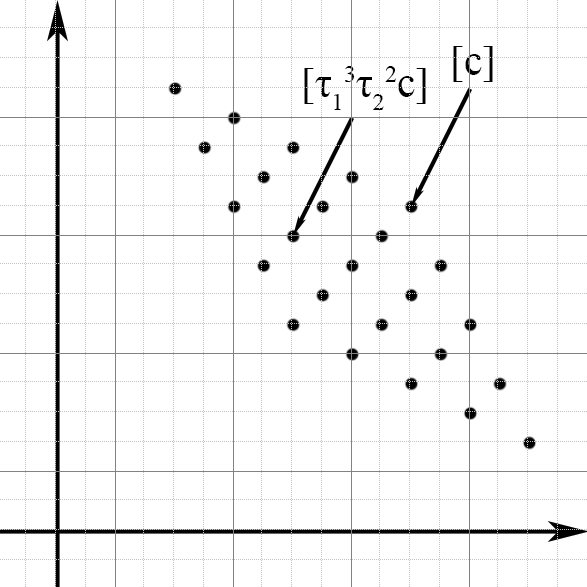}
\caption{The ``parallelepiped.''}
\label{fig:pic}
\end{figure}

We choose $R$ to be the set of all recurrent configurations of ASM on $G$, whose threshold on each vertex is set uniformly to $T$, rather than the conventional $\deg v$. For convenience, let us temporarily write
\begin{equation*}
S(R) := \frac{1}{|R|}\sum_{c \in R} S(c).
\end{equation*}

We will show that $S(R)^\circ$ is the steady state. 

By Proposition \ref{prop:assoc}, it suffices to show that $(S(R) + v_i)^\circ = S(R)^\circ$ for all $i = 1, \ldots, L$. By the theory of ASM, $R + v_i = \{c + v_i: c \in R\}$ contains exactly one representative of each of the equivalence classes of configurations i.e. for each $c' \in R + v_i$, there exists exactly one $c \in R$ such that $c'$ stabilizes to $c$ in ASM. By \eqref{eq:commute}, $S(c')^\circ = S(c)^\circ$, and therefore, by Proposition \ref{prop:assoc} again,
\begin{equation*}
(S(R) + v_i)^\circ =\left(\frac{1}{|R|}\sum_{c' \in R+v_i} S(c')\right)^\circ = \left(\frac{1}{|R|}\sum_{c' \in R+v_i} S(c')^\circ\right)^\circ = \left(\frac{1}{|R|}\sum_{c \in R} S(c)^\circ\right)^\circ = S(R)^\circ,
\end{equation*}
as desired.
\end{proof}

The point of the following proposition is that, in a SSP, if we start from any configuration, and repeat sufficiently many times the Markov process of adding a sand grain at random --- e.g. add $\frac{1}{L}\sum v_i$ --- and re-stabilizing, then we obtain a configuration that is arbitrarily close to the steady state. In the context of LLL, in place of the Markov chain one takes a large input, which is why the proposition below is stated the way it is.

\begin{proposition} \label{prop:conv_steady}
We continue with the notations of Theorem \ref{thm:steady}. In addition, below we interpret a pure configuration as an $L$-dimensional vector with integer entries, and a mixed configuration as a finitely supported function on the set of pure configurations. Then we impose the typical Euclidean metric on the space of pure or mixed configurations accordingly.

For any $\varepsilon > 0$, there exists $D > 0$ such that the stabilization $c^\circ$ of any non-negative pure configuration $c$ whose distance from origin is greater than $D$ is within an $\varepsilon$ distance of $S(c')^\circ$ with respect to the uniform norm, where $c'$ is the pure configuration that is recurrent and is equivalent to $c$ under the ASM on $G$.
\end{proposition}
\begin{proof}
If $c$ is far enough from the origin, then it must be toppled at each vertex arbitrarily many times in order for it to become stabilized. So without loss of generality, assume that $c$ can be legally toppled at every vertex. Write $V \backslash \{s\} = \{v_1, v_2, \ldots, v_L\}$ as earlier, and consider the configuration
\begin{equation} \label{eq:oinkoink}
\left( \prod_{i=1}^L T_{v_i} \right) c = \sum_{k_1=0}^{n-1} \cdots \sum_{k_L=0} \left(\prod_{j=1}^L \wp(k_j)\tau_j^{k_j} \right) [c].
\end{equation}

\eqref{eq:oinkoink}, like $S(c)$, may be thought of as a distribution on a parallelepiped illustrated in Figure \ref{fig:pic}. By Lemma \ref{lemma:steady}, if we ``push'' the parallelepiped --- in the manner similar to \eqref{eq:commute} --- from every direction sufficiently many times, the weight on each and every point will be arbitrarily close to $\prod_{j=1}^L r(k_j)$. Hence, if $c$ is sufficiently far enough from the origin, the stabilization of \eqref{eq:oinkoink} will be arbitrarily close to $S(c')$. This completes the proof.
\end{proof}

The steady state alone forms the minimal ideal of the commutative monoid $M$ of the nonnegative stable configurations of SSP. One may think this is disappointing, since its counterpart in ASM has a very rich theory. However, the theory of the abelian sandpile group can be extended to SSP in a straightforward way, if phrased in terms of the cosets of the ``translations'' $< \tau_i  : i = 1, \ldots, L>$, since SSP respects this structure. Indeed, in our proof of Theorem \ref{thm:steady}, it can be seen that the underlying ASM theory plays a role. Our SSP theory is building up on top of the ASM theory, rather than replacing it.


\subsection{The shape of the steady state on $A_L$}

Ideally, we would like to prove the various quantitative properties of the steady state that we observe in Figure \ref{fig:lllssp}: the middle values are almost identical to one another, but near the boundaries there is a sharp drop in pile height, et cetera. This can be interpreted as the following combinatorial problem. Consider the parallelepiped in Figure \ref{fig:pic}, and suppose $[c]$ is a stable configuration; it may help to think of $[c]$ as the maximal stable configuration, i.e. has height $T-1$ on all sites. Then some configurations on the parallelepiped are stable e.g. $[\tau_1^3\tau_2^2c]$, but others may not be. Because $[c]$ is stable, we cannot legally push the entire face of the parallelepiped as we have done in the proofs of Theorem \ref{thm:steady} and Proposition \ref{prop:conv_steady}. Instead, we must ``fold'' the parts of the parallelepiped that are outside the set of stable configurations. Understanding the steady state then amounts to understanding the resulting distribution on the ``folded'' parallelepiped --- they are identical. In some sense, this parallelepiped-folding replaces the role of the burning algorithm (see \cite{Dhar90}, also \cite{CMS13}) in the study of the steady state.

The difficulty lies in describing this folded distribution in a mathematically concise manner. The best we can prove at the moment is that there exists a gap, whose size is independent of the system size $L$, between the average-case and the worst-case RHF \eqref{eq:rhf} of SSP on $A_L$. This is in analogy with the folklore observation that there exists a gap between the average-case and the worst-case RHF of the LLL algorithm.

\begin{proposition} \label{prop:rhf}
Let $C$ be as in the proof of Theorem \ref{thm:steady}. Then for any $\varepsilon > 0$ and $L$ sufficiently large, the steady state of SSP on $A_L$ has the average RHF bounded from above by $T/2 - C^{-1}/e^2 + \varepsilon$.
\end{proposition}

\begin{remark}
Note that the maximum $\log(\mathrm{RHF})$ equals $(T-1)L/2(L+1) \approx T/2$. For the SSP in the introduction i.e. $\wp =$ (uniform distribution on $\{1, \ldots, 2I-1\}$), $C^{-1} \approx I$ so the proposition yields the bound $T/2 - I/e^2$, whereas experimentally we have $\approx T/2 - I/4$. 
\end{remark}

\begin{proof}[Proof of Proposition \ref{prop:rhf}]

As in the remark above, the greatest possible $\log(\mathrm{RHF})$ equals $M := (T-1)L/2(L+1)$. For $a > 0$, we will estimate the number of pure stable configurations whose RHF is greater than $M - a$. If $c$ satisfies such a condition, then writing $s_i = T - c(i)$, from \eqref{eq:rhf} it follows that
\begin{equation} \label{ribbit}
a > \frac{1}{(L+1)^2}\sum_{i=1}^L (L+1-i)s_i.
\end{equation}

Moreover, this is an if-and-only-if condition. Hence it comes down to measuring the volume of the set of vectors $(s_1, \ldots, s_L)$, $s_i > 0$, such that \eqref{ribbit} holds. This set forms a simplex, with one vertex at origin, and $L$ other vertices given by
\begin{equation*}
\left(0, \ldots, \frac{a(L+1)^2}{L+1-i}, \ldots, 0\right)
\end{equation*}
for each $i= 1, \ldots, L$, whose only non-zero entry is the $i$-th entry. The volume of this simplex equals
\begin{equation*}
\frac{a^L(L+1)^{2L}}{(L!)^2},
\end{equation*}
which, by Stirling's formula, is strictly bounded by $a^Le^{2L}$ for all sufficiently large $L$. Furthermore, the maximum possible probability mass of the steady state is given by $C^{L}$. Hence, for $L$ large, the portion of the steady state lying on the set $(\mathrm{RHF}) > M - a$ is arbitrary small if
\begin{equation*}
a < C^{-1}/e^2.
\end{equation*}

This completes the proof.
\end{proof}

\section{Behavior of SSP on one-dimensional lattices}

In this section, we present a pilot numerical study on the behavior of SSP pertaining to quantitative properties of its steady state, whose existence is established in the previous section. We restrict our attention to SSP on one-dimensional grids $A_L$ for $L = 400, 4000, 40000$, with $T = 400$ and $\wp$ the uniform distribution on $\{1, \ldots, 100\}$, adding $50$ grains to a random site for each step in the associated Markov chain. We hope to be able to carry out more extensive experiments at a later time.


\subsection{Avalanche size distribution}

\begin{figure}
\includegraphics[scale = 0.5]{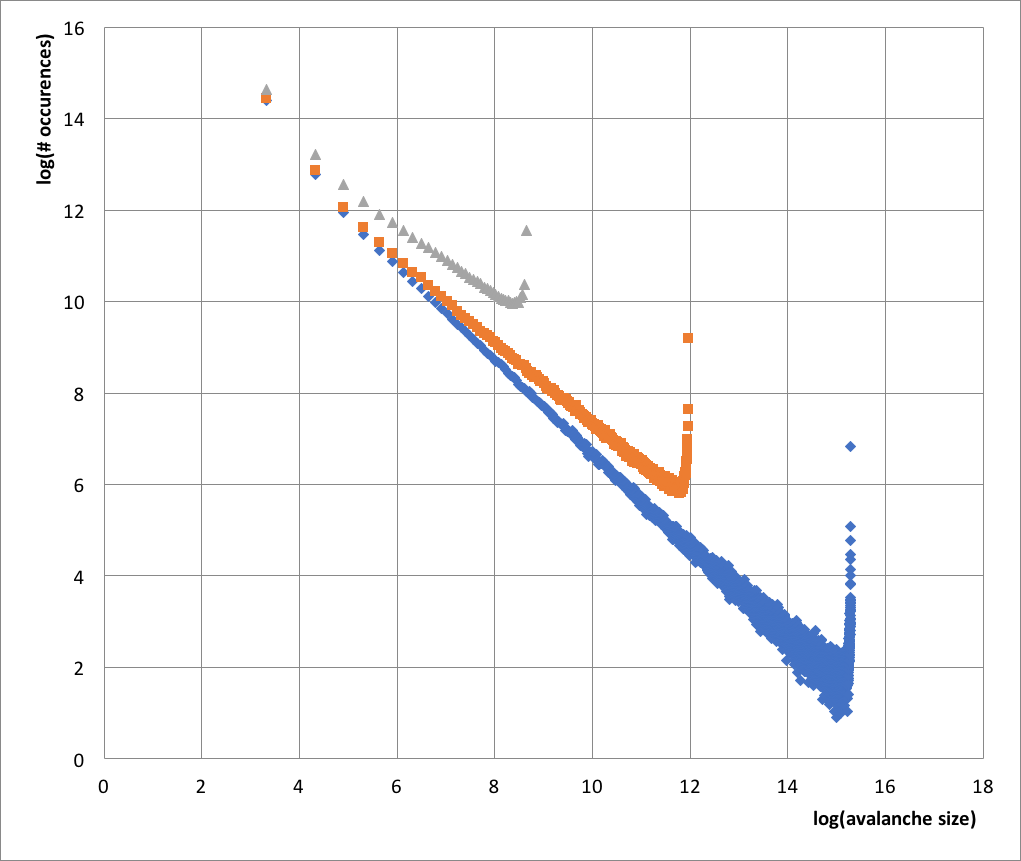}
\caption{The avalanche size distribution of SSP. Gray, orange, blue graphs represent cases $L = 400, 4000, 40000$ respectively.}
\label{fig:socssp} 
\end{figure}


Figure \ref{fig:socssp} presents the log-log scale graph of the avalanche size --- i.e. the number of sites on which at least one toppling occurred during a stabilization --- distribution of SSP. For each system size $L =400, 4000, 40000$, we made one million trials. The $x$-axis represents $\log_2$ of avalanche sizes, and the $y$-axis represents $\log_2$ of the number of occurences.

Figure \ref{fig:socssp} clearly suggests a power law, plus a delta distribution near $\log_2 L$. The data points for $L = 40000$, excluding those at the tail, form a line of slope $\approx -0.98$. Of course, more experiments are needed to precisely determine the exponent of SSP.

It may be amusing to recall the behavior of one-dimensional ASM, which is well-known to behave in the totally opposite way, i.e. (avalanche size) $\propto$ (frequency). On the other hand, several stochastic models in dimension 1 are known to demonstrate a power law (\cite{GDM16}, \cite{HPC11}, \cite{SD09}).


\subsection{Average heights of the piles in the steady state}

\begin{figure}
\includegraphics[scale = 0.35]{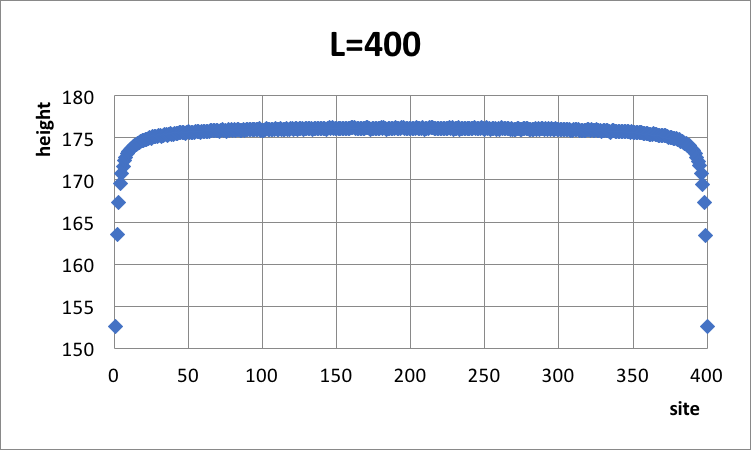}
\includegraphics[scale = 0.35]{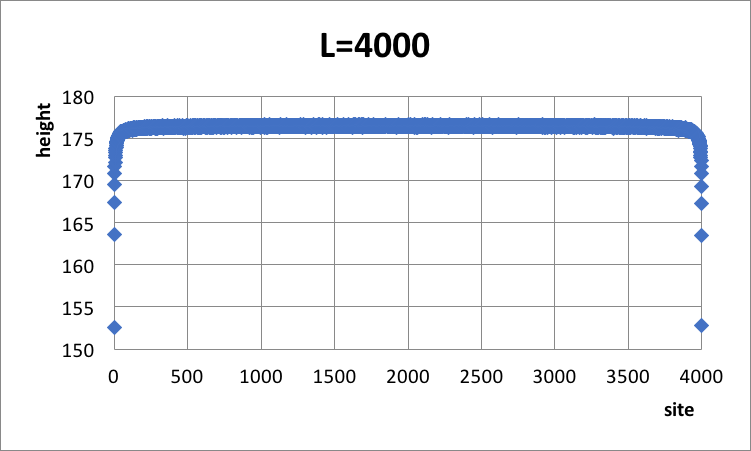}
\includegraphics[scale = 0.35]{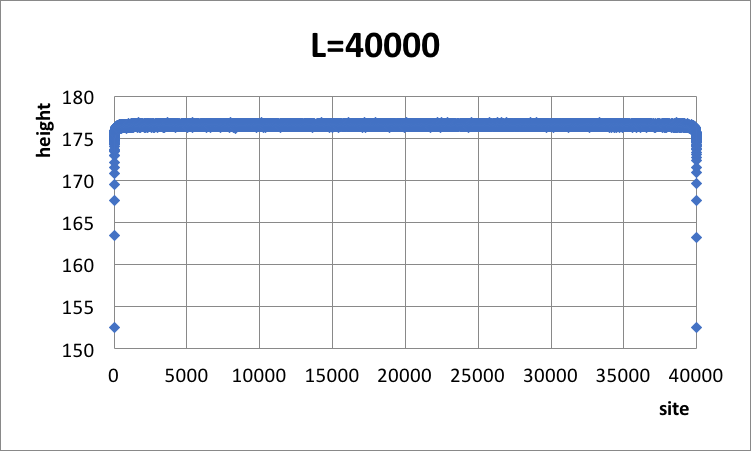}
\caption{The average steady state density of SSP.}
\label{fig:shapessp} 
\end{figure}

\begin{figure}
\includegraphics[scale = 0.4]{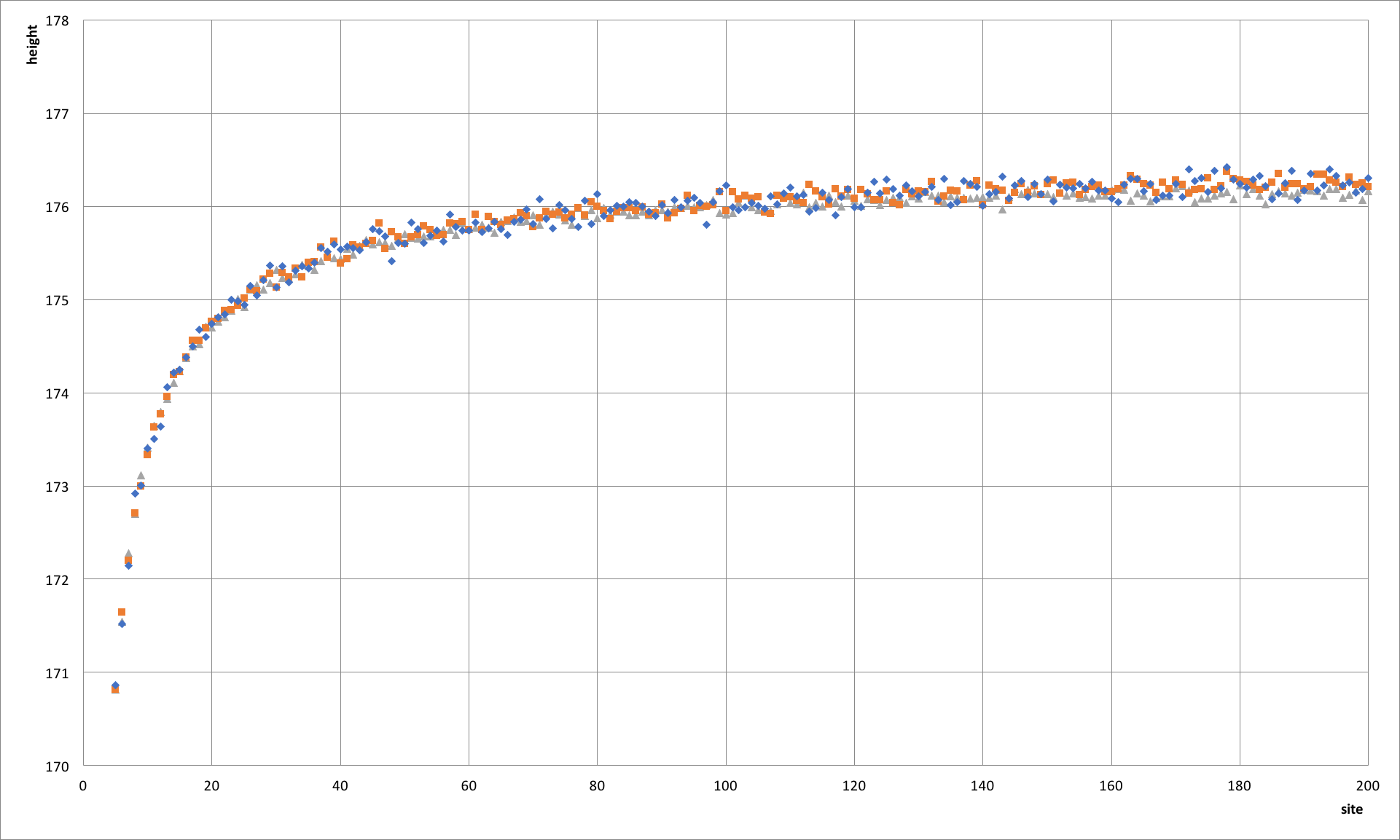}
\caption{The average steady state density of SSP for different system sizes. Gray, orange, blue graphs represent cases $L = 400, 4000, 40000$ respectively.}
\label{fig:shapeover} 
\end{figure}

Figure \ref{fig:shapessp} presents the (empirical) average steady state density of SSP for system sizes 400, 4000, 40000, respectively. The values in the middle seem to stabilize at $\approx 176.6$, and in the extremes at $\approx 152.5$.

The obvious curiosity is the fall-off shape at the boundaries. Figure \ref{fig:shapeover} puts together the three graphs in Figure \ref{fig:shapessp} for a comparison. Here we can see that the boundary areas overlap, which indicate that the effect of a boundary is independent of system size. Also observe that, in Figure \ref{fig:shapeover}, the first half of the average output shape of $L = 400$ overlaps with the other shapes; the latter half, of course, must be the mirror image of the first half. Indeed, the first half of the average output shape of $L = 4000$ is also found to overlap with that of $L = 40000$.

At this point, we are able to only speculate as to why the average height diminishes near the boundaries, as follows. When a toppling occurs at one of the middle vertices, it is expected that some of the toppled sand comes back when the neighboring vertex topples. Some grains of sand may travel several steps away from where it started and then come back. However, when a toppling occurs at the left extreme, for example, the sand that moved to the left is lost forever.

Other stochastic sandpile models in dimension one, such as the abelian Manna model (\cite{SD09}) and the Oslo rice-pile model (\cite{GDM16}), have also been observed to exhibit similar-looking shapes. \cite{GDM16} attributes it to the finite-size scaling theory --- see equation (7) in \cite{GDM16}. However, in the SSP case, the same explanation seems to account for only the first 10 points or so away from the boundary.

It is also natural to ask whether the general features of the average steady state density are preserved if $\wp$ is replaced with something other than a uniform distribution. To this end, we ran several brief experiments with the Poisson distribution, and also the family of fat-tailed distributions of index $\alpha$ i.e. $\mathrm{Pr}(X > x) \sim x^{-\alpha}$, which is often used to produce unusual outcomes.\footnote{These distributions have infinite support, and thus they are not strictly covered by the theory developed here. However, the parallelepiped-pushing idea still goes through.} In all cases, we still observe both the flatness of the middle values and the diminishing amount of sand near the boundaries, suggesting that these are general properties of one-dimensional stochastic sandpiles. For instance, see Figure \ref{fig:shapeft}, the outcome for the case of fat-tailed distributions of indices $3$ and $4$, with mean $\approx 15, 13$ and standard deviation $\approx 17, 14$ respectively. It may be interesting to study how the average output shape of SSP and the representative values of $\wp$ are related; for example, it seems reasonable to conjecture that the mean of $\wp$ correlates with the pile height in the middle, and the standard deviation with the boundary effect.

\begin{figure}
\includegraphics[scale = 0.4]{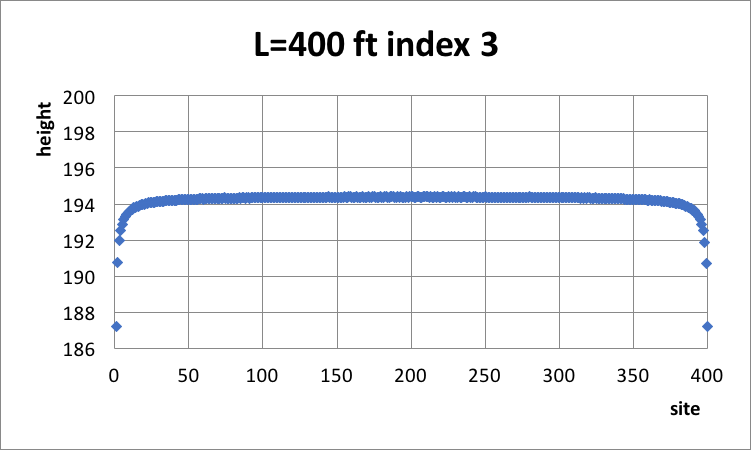}
\includegraphics[scale = 0.4]{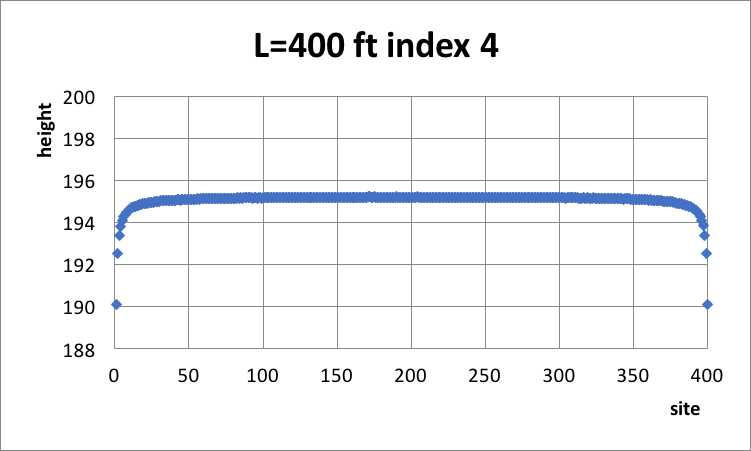}
\caption{The average steady state density of SSP, with respect to fat-tailed distributions of indices $3$ and $4$.}
\label{fig:shapeft} 
\end{figure}







\end{document}